\newif\ifpublic
\publictrue 

\newif\ifcomment

\documentclass[11pt]{article}

\usepackage{ttpalatino}
\usepackage[margin=1in]{geometry}

\usepackage{amsmath}
\usepackage{amssymb}
\usepackage{amsfonts}
\usepackage{graphicx}
\usepackage{fullpage}
\usepackage{pdfsync}
\usepackage{underscore}  
\usepackage{amsthm}
\usepackage{xcolor}

\usepackage[bookmarks,colorlinks,breaklinks]{hyperref}  
\hypersetup{linkcolor=blue,citecolor=blue,filecolor=blue,urlcolor=blue} 
\newcommand{\lref}[2][]{\hyperref[#2]{#1~\ref*{#2}}}
\renewcommand{\eqref}[1]{\hyperref[#1]{(\ref*{#1})}}
\numberwithin{equation}{section}

\theoremstyle{plain}
\newtheorem{lem}{Lemma}[section]
\newtheorem{theorem}[lem]{Theorem}
\newtheorem{lemma}[lem]{Lemma}

\newtheorem{conjecture}[lem]{Conjecture}
\newtheorem{claim}[lem]{Claim}

\newtheorem{definition}[lem]{Definition}

\theoremstyle{definition}

\ifpublic
\usepackage[disable]{todonotes}
\else
\usepackage[colorinlistoftodos]{todonotes}
\fi

\newcommand{\girish}[1]{\todo[color=green,size=\footnotesize]{G: #1}}

\DeclareMathOperator{\poly}{poly}
\DeclareMathOperator{\rank}{rank}

\DeclareMathOperator{\Term}{Term}

\DeclareMathOperator*{\E}{\mathbb{E}}

\newcommand{\TSAT}{3\text{-}\mathsf{SAT}}
\renewcommand{\th}{\textsuperscript{th}~}

\newcommand{\etal}{{\em et.~al.}}
\newcommand{\F}{\mathbb{F}}

\title{Reducing uniformity in Khot-Saket hypergraph coloring
  hardness reductions}

\author{
Girish Varma\thanks{Tata Institute of Fundamental Research,
   India. Supported by Google India under the Google India PhD
   Fellowship Award. Email : girishrv@tifr.res.in}
 }

\date{\today}

\begin{document}
\maketitle

\begin{abstract}
  In a recent result, Khot and Saket [FOCS 2014] proved the
  quasi-NP-hardness of coloring a $2$-colorable $12$-uniform hypergraph
  with $2^{{(\log n)}^{\Omega(1)}}$ colors. This result was proved
  using a novel outer PCP verifier which had a strong soundness
  guarantee. We  reduce the arity
  in their result by modifying their 12-query inner verifier to an
  8-query inner verifier based on the hypergraph coloring hardness
  reductions of Guruswami~\etal~[STOC 2014]. More precisely, we prove
  quasi-NP-hardness of the following problems on $n$-vertex
  hypergraphs.
\begin{itemize}
\item Coloring a $2$-colorable $8$-uniform hypergraph with $2^{(\log n)^{\Omega(1)}}$ colors.
\item Coloring a $4$-colorable $4$-uniform hypergraph with $2^{(\log n)^{\Omega(1)}}$ colors.
\end{itemize}
\end{abstract}

\section{Introduction}

The discovery of the low-degree long code aka short code by
Barak~\etal~\cite{BarakGHMRS2012} has over the last one year led to a
sequence of results improving our understanding of the hardness of
constant colorable hypergraphs with as few colors as possible. A
$k$-uniform hypergraph is a collection of vertices and hyperedges such
that every hyperedge is a subset of $k$ vertices. A hypergraph is
said to be $q$-colorable if the vertices of the hypergraph can be
colored with at most $q$ colors such that no hyperedge is
monochromatic. An independent set in a hypergraph is a collections of
vertices such that no hyperedge is wholly contained within the
collection. Note that a hypergraph is $q$-colorable iff it the set of
vertices can be partitioned into at most $q$ independent sets. 

Prior to the low-degree long code, all hardness reductions for
hypergraph coloring where proved using the long code~\cite{BellareGS1998}
which resulted in a huge disparity between the best known positive and
negative results for hypergraph coloring: the best known approximation
algorithms require at least $n^{\Omega(1)}$ colors to color a constant
colorable (hyper)graph while the inapproximability results could only
rule out at best $(\log n)^{O(1)}$ colors. The situation was redeemed by
introduction of the low-degree long code, a derandomization of the
long code, which was then adapted by Dinur and Guruswami~\cite{DinurG2013}
toward proving inapproximability results. Building on the
Dinur-Guruswami framework, Guruswami~\etal~\cite{GuruswamiHHSV2014}
showed that it is quasi-NP-hard to color a 2-colorable 8-uniform
hypergraph with $2^{2^{\Omega(\sqrt{\log \log n})}}$ colors. Both the
Dinur-Guruswami and Guruswami~\etal~results were obtained by modifying
the innermost PCP verifier to work with the low-degree long
code. Shortly thereafter, in a remarkable improvement,
Khot and Saket~\cite{KhotS2014b} showed that it is
quasi-NP-hard to color a 2-colorable 12-uniform hypergraph with
$2^{(\log n)^{\Omega(1)}}$ colors.  They obtained this result by using
an 12-query inner PCP verifier based on the quadratic code, ie., a low-degree
long code with degree two. However, to use a quadratic code based
inner verifier, they needed an outer PCP verifier with a significantly
stronger soundness guarantee than the standard outer PCP verifier
obtained from parallel repetition of the PCP Theorem. In particular,
they needed an outer PCP verifier, which in the soundness case, would
not be satisfied by a short list of proofs even in {\em
  superposition}\footnote{We will not require the exact definition of
  {\em satisfying in superposition} for this note. See
  \lref[Theorem]{thm:quad-label-cover} for the details of the Khot-Saket outer PCP
  verifier.}. The construction of this outer PCP verifier with this
stronger soundness guarantee is the main technical ingredient in the
result of Khot
and Saket~\cite{KhotS2014b}. We show that this outer PCP verifier of
Khot and Saket can in fact be combined with a 8-query inner PCP
verifier based on the Guruswami~\etal~inner PCP verifier to obtain a
hardness result for 2-colorable 8-uniform hypergraphs. More precisely, we show the following.

\begin{theorem}\label{thm:2c8u}
For every constant $\epsilon>0$ there is a quasi-polynomial time reduction from $\TSAT$ to a $8$-uniform hypergraph $\mathcal G$ on $n$ vertices such that,
\begin{enumerate}
\item YES Case:  If the $\TSAT$ instance is satisfiable then $\mathcal G$ is $2$-colorable.
\item NO Case: If the $\TSAT$ instance is unsatisfiable then $\mathcal G$ does not have an independent set of relative size $2^{-(\log n)^{\frac{1}{20} -\epsilon}}$.
\end{enumerate}
\end{theorem}

Guruswami~\etal~\cite{GuruswamiHHSV2014} also proved how to reduce the
uniformity in certain reductions from 8 to 4 at the cost of increasing
the number of colors from 2 to 4. We note that a similar trick can be
performed in our setting to obtain the following result.

\begin{theorem}\label{thm:4c4u}
For every constant $\epsilon>0$ there is a quasi-polynomial time reduction from $\TSAT$ to a $4$-uniform hypergraph $\mathcal G$ on $n$ vertices such that,
\begin{enumerate}
\item YES Case:  If the $\TSAT$ instance is satisfiable then $\mathcal G$ is $4$-colorable.
\item NO Case: If the $\TSAT$ instance is unsatisfiable then $\mathcal G$ does not have an independent set of relative size $2^{-(\log n)^{\frac{1}{20} -\epsilon}}$.
\end{enumerate}
\end{theorem}

We remark that the analyses of the inner verifier in both the above
theorems is simpler than the analyses of the corresponding inner
verifiers in Guruswami~\etal~and Khot-Saket results. Furthermore,
in the language of covering complexity~\footnote{The covering number
  of a CSP is the minimal number of assignments to the vertices so
  that each hyperedge is covered by at least one assignment.} introduced by Guruswami, H{\aa}stad and Sudan~\cite{GuruswamiHS2000}, (the
proof of) \lref[Theorem]{thm:4c4u} demonstrates a Boolean 4CSP for
which it is quasi-NP-hard to distinguish between covering number of 2
vs. $(\log n)^{\Omega(1)}$.

\section{Preliminaries}
Let $\F$ denote the field $GF(2)$. 
Let $\F^{m\times m}$ be the vector space of $m\times m$ matrices over
the field $\F$. Our inner verifier is based on the quadratic code, which is a specialization of the low degree long code to degree $2$.
\begin{definition}[Quadratic Code]
The quadratic code of $x \in \F^m$ is a function $A_x:\F^{m\times m} \rightarrow \F$ defined as $A_x(X):= \langle X, x\otimes x \rangle$. 
\end{definition}

Our reductions makes use of the following outer PCP verifier of Khot
and Saket~\cite{KhotS2014b}. As stated in the introduction, these
instances have stronger soundness conditions which make them amenable
for composition with a quadratic code based inner verifier.
\begin{theorem}[{\cite[Theorem~7.2]{KhotS2014b}}]
\label{thm:quad-label-cover}
There is a quasi-polynomial time reduction from an instance of $\TSAT$ to a bi-regular instance $(U,V,E,\Pi)$ of Label Cover such that 
\begin{itemize}
\item Vertex sets $U$ and $V$ are bounded in size by $N$.
\item The label sets are $\F_2^{r\times r},\F_2^{m\times m}$ for $U$ and $V$ respectively.
\item For $e \in E$, the map $\pi^e:\F_2^{m\times m} \rightarrow \F_2^{r\times r}$ is a linear transformation that maps symmetric matrices to symmetric matrices\footnote{The property that $\pi$ maps symmetric matrices to symmetric matrices is easy to see from the proof of {\cite[Theorem~7.2]{KhotS2014b}}.}. For an $r\times r$ matrix $X$, $X\circ \pi^e$ is the unique $m \times m$ matrix such that $\langle X \circ \pi^e, Y \rangle = \langle X , \pi^e Y \rangle$.
\item For each vertex $v \in V$, there is a constraint $C_v$ that is a a conjunction of homogeneous linear equations on the entries of the $m \times m$ matrix label.
\item $\delta \leq  2^{- \log^{1/3} N}$ and $k \geq  (\log N)^{1/9}$.
\end{itemize}
The reduction satisfies:
\begin{enumerate}
\item Completeness : If the $\TSAT$ instance is satisfiable then there is a labeling $x_u \otimes x_u$ for $u \in U$ and $y_v \otimes y_v$ for $v\in V$ such that
\begin{itemize}
\item  for each $v \in V$, $y_v \in \F_2^m$ has the $m$\th coordinate $1$ and $y_v\otimes y_v$ satisfies the constraint $C_v$,
\item  for each $(u,v) \in E$, $\pi_{u,v}(y_v \otimes y_v) = x_u \otimes x_u$.
\end{itemize}

\item Soundness : If the $\TSAT$ instance is not satisfiable then the following cannot hold: There are symmetric matrices $M_u \in \F_2^{r\times r}, M_v \in \F_2^{m\times m}$ for $u\in U, v\in V$ of rank $\leq k$ such that 
\begin{itemize}
\item  for each $v \in V$, $M_v \in \F_2^{m \times m}$ has the $(m,m)$\th coordinate $1$ and $M_v$ satisfies the constraint $C_v$,
\item for $\delta$ fraction of edges $e$, $\pi_e(M_v) = M_u$.
\end{itemize}

\item Smoothness : For any $v \in V$ and any symmetric  non-zero matrix $M_v$ with rank $\leq k$, over a random choice of an edge $e$ incident on $v$,
$$ \Pr[\pi_e(M_v) = 0] \leq \delta/2.$$

\end{enumerate}
\end{theorem}
\ifcomment
For our reduction in the case of $3$-uniform hypergraphs, we need to construct a multi-layered label cover from the bipartite instance described above. We also need all the matrix labels to be over the field $\F_3$.
Multilayered label cover has been used widely used for proving hypergraph coloring hardness results in \cite{}. 

\begin{definition}[Multi-layered Label Cover]
 The parameter $\ell$ will denote the number of layers in the multi-layered label cover that is constructed. Let $(U,V,E,\Pi)$ be an instance of the bipartite label cover with the matrices in the label set over $\F_3$.
\begin{itemize}
\item For $0 \leq i < \ell$, the vertices in the $i$th layer are $V_i = U^{i}\times V^{\ell - i}$. 
\item The label set $\mathcal{L}_i$ for $V_i$ is  $(\F_3^{r\times r})^i \times (\F_3^{m\times m})^{\ell - i}$. We will think of $\mathcal{L}_i$ as the set of block diagonal matrices having dimension $m_i :=i\times r+ (\ell - i) m$, with the previously mentioned matrices in the diagonals.
\item 
  There is an edge between $u=(u_1,\cdots, u_\ell) \in V_i$ and $v =(v_1,\cdots, v_\ell)\in V_j$, if $(u_k,v_k)\in E$  for $k \in \{i+1, \cdots, j\} $ and $u_k = v_k$ otherwise.
  
\item 
  The projection constraint is a product of constraints for each coordinate. For $k \in \{i+1, \cdots, j\}$, it is the constraint  $\pi_{u_k,v_k} \in \Pi$ corresponding to the edge $(u_k,v_k)\in E$ and identity otherwise.
  
\end{itemize} 
\end{definition}

\begin{theorem}[$\ell$-layered Quadratic Label Cover over $\F_3$]
The construction of the $\ell$-layered Label cover satisfies:
\begin{enumerate}
\item Completeness : Given a labeling for the bipartite instance  $x_u \otimes x_u$ for $u \in U$ and $y_v \otimes y_v$ for $v\in V$ that satisfies all the constraints, the labeling $(x_{u_1} \otimes x_{u_1},\cdots,x_{u_i} \otimes x_{u_i}, y_{v_{i+1}} \otimes y_{v_{i+1}}, \cdots, y_{v_\ell} \otimes y_{v_\ell})$ for the vertex $(u_1,\cdots, u_i,v_{i+1},\cdots, v_\ell) \in V_i, \forall i \in \{0,\cdots \ell -1\}$ satisfies all the constraint in the multi-layered label cover instance.

\item Soundness : If the bipartite label cover was obtained from a unsatisfiable instance of $\TSAT$ then in the multi-layered label cover for every $0\leq i < j <\ell$, the following cannot hold: There are symmetric matrices $M_u \in \mathcal{L}_i, M_v \in \mathcal{L}_j$ for $u\in V_i, v\in V_j$ of rank $\leq k$ such that 
\begin{itemize}
\item  for any vertex $u=(u_1,\cdots,u_i,v_{i+1},\cdots,v_\ell)$, the diagonal blocks in $M_u$ corresponding to $v_{i+1},\cdots,v_\ell$ have the $(m,m)$th coordinate $1$ and satisfies the constraints $C_{v_i}$. 
\item for $\delta$ fraction of edges $e=(u,v)$, $\pi_e(M_v) = M_u$.
\end{itemize}

\item Smoothness : For any $u \in V_i$ and any symmetric  non-zero matrix $M_v \in \mathcal{L}_i $ with rank $\leq k$, over a random choice of an edge $e$ incident on $u$,
$$ \Pr[\pi_e(M_v) = 0] \leq \delta/2.$$

\item Weakly Dense : For any $\delta >0$, given $\ell' \geq 2/\delta$ layers $l_1 < l_2 < \cdots < l_{\ell'}$  and given any sets $S_i \subseteq V_{l_i}$ with $|S_{l_i}| \geq \delta |V_{l_i}|$, there always exists two layers $i$ and $j$ such that the edges between $S_i$ and $S_j$ is at least $\delta^2/4$ fraction of the edges between the layers.

\end{enumerate}
\end{theorem}
\fi

\section{$2$-colorable $8$-uniform hypergraphs}
In this section we prove \lref[Theorem]{thm:2c8u}. Our reduction starts from the label cover instances given by \lref[Theorem]{thm:quad-label-cover}. Let $(U,V,E,\Pi)$ be an instance of the label cover. We will construct a hypergraph $\mathcal G=(\mathcal V, \mathcal E)$. For $v\in V$, let $\mathcal H_v \subseteq \F_2^{m\times m}$ be the dual of the subspace of the set matrices that are symmetric and which satisfies the constraint $C_v$. The set of vertices $\mathcal V$ will be the same as $V\times\left( \F_2^{m\times m}/\mathcal H_v\right)$. Any $2$-coloring of $\mathcal G$ is a collection of functions $A'_v:\F_2^{m\times m}/\mathcal H_v \rightarrow \{0,1\}$ for $v\in V$. For any such function, we can uniquely extend it to get $A_v:\F_2^{m\times m} \rightarrow \{0,1\}$ which is constant over cosets of $\mathcal H_v$. This method is called folding and it ensures that $A_v$ satisfies the following: if $\alpha \in \F_2^{m\times m}$ is such that $\widehat A_v(\alpha)$ is non-zero, then $\alpha$ is symmetric and satisfies $C_v$. 
\paragraph{}
The set of edges $\mathcal E$ will be defined by the test mentioned below, which checks whether a supposed $2$-coloring $A'_v:\F_2^{m\times m}/\mathcal H_v \rightarrow \{0,1\}$ is valid. There is an edge in $\mathcal E$ between any set of vertices in $\mathcal V$ that are queried together by the test. The test will be querying the extended functions $A_v$ at matrices in $\F_2^{m\times m}$ instead of $A'_v$. So a query to $A_v$ at $X \in \F_2^{m\times m}$ corresponds to a query to $A'_v$ at the coset of $\mathcal H_v$ that contains $X$.
\paragraph{$2$-Colorable $8$-Uniform Test $\mathcal T_{2,8}$}
\begin{enumerate}
\item Choose $u \in U$ uniformly at random and $v,w \in V$ uniformly and independently at random from the neighbors of $u$. Let $\pi,\sigma:\F_2^{m \times m} \rightarrow \F_2^{r \times r}$ be the projections corresponding to the edges $(u,v),(u,w)$ respectively. Uniformly and independently at random choose $X_1,X_2, Y_1,Y_2 \in \F_2^{m \times m}$ and $\overline x, \overline y, \overline z , \overline x', \overline y', \overline z' \in \F_2^m$ and $F \in \F_2^{r\times r}$. Let $\overline e_m \in \F_2^m$ be the vector with only the $m$\th entry $1$ and the rest is $0$.
\item Accept if and only if the following $8$ values are not all equal :
\begin{align*}
&A_v(X_1)&&A_v(X_3) &\text{ where } X_3 &:= X_1+\overline x \otimes \overline y +F\circ \pi\\
&A_v(X_2)& &A_v(X_4) &\text{ where } X_4 &:= X_2 + (\overline x + \overline e_m) \otimes \overline z  + F \circ \pi\\
&A_w(Y_1)&&A_w(Y_3) &\text{ where } Y_3 &:= Y_1+\overline x'\otimes \overline y'+F\circ \sigma + \overline e_m \otimes \overline e_m\\
&A_w(Y_2)&&A_w(Y_4) &\text{ where } Y_4 &:= Y_2 + (\overline x'+\overline e_m)\otimes \overline z'  +F \circ \sigma + \overline e_m \otimes \overline e_m
\end{align*}
\end{enumerate}

\subsection{YES Case}\label{sec:yes-28}
Let  $\overline y_v \otimes \overline y_v$ for $v\in V$ and $\overline x_u \otimes \overline x_u$ for $u\in U$ be a perfectly satisfying labeling of the label cover instance. That is, for every $(u,v)\in E, \pi_{u,v}(\overline y_v\otimes\overline y_v) = \overline x_u \otimes \overline x_u$. Such a labeling is guaranteed by the YES instance of label cover, with the additional property that the $m$\th coordinate of $\overline y_v$ is $1$. Consider the following $2$-coloring of $\mathcal G$: for each $v \in V$,  $A_v(X):= \langle X,\overline y_v \otimes \overline y_v \rangle$. Note that such a function is constant over cosets of $\mathcal H_v$.  Let 
\begin{align*}
&x_1 := \langle X_1 , \overline y_v \otimes \overline y_v \rangle& ~~~~~&x_2 := \langle X_2 , \overline y_v \otimes \overline y_v \rangle\\
&y_1 := \langle Y_1 , \overline y_w \otimes \overline y_w \rangle& ~~~~~&y_2 := \langle Y_2 , \overline y_w \otimes \overline y_w \rangle
\end{align*}
and $f := \langle F, \overline x_u \otimes \overline x_u \rangle$.  Note that $\langle F\circ \pi_{u,v}, \overline y_v\otimes \overline y_v \rangle = \langle F, \pi_{u,v}(y_v \otimes y_v)\rangle = \langle F, \overline x_u \otimes \overline x_u\rangle$ , and $\langle \overline e_m \otimes \overline e_m ,\overline y_v \otimes \overline y_v \rangle = \langle \overline e_m , \overline y_v \rangle = 1$. Using these, the assignments to the $8$ query locations are:
\begin{align*}
&x_1& ~~~~~&x_1+\langle \overline y_v, \overline x\rangle \langle \overline y_v, \overline y \rangle + f\\
&x_2& ~~~~~&x_2 + \left(\langle\overline y_v, \overline x\rangle +1\right) \langle \overline y_v, \overline z\rangle + f \\
&y_1& ~~~~~&y_1+\langle \overline y_w, \overline x'\rangle \langle\overline y_w, \overline y' \rangle  + f + 1\\
&y_2& ~~~~~&y_2 + \left(\langle\overline y_w, \overline x'\rangle +1\right) \langle\overline y_w, \overline z'\rangle + f + 1
\end{align*}
It is easy to see at least one of the $4$ rows are always not equal. Hence $A$ is a valid $2$-coloring of $\mathcal G$.
 
\subsection{NO Case}\label{sec:no-28}
Suppose the reduction was applied to a NO instance of label cover. Let $k$ and $\delta$ be the parameters specified by \lref[Theorem]{thm:quad-label-cover}.

\begin{lemma}\label{lem:soundness-28}
If there is an independent set in $\mathcal G$ of relative size $s$ then
$$ s^8 \leq \delta +\frac{1}{2^{k/2+1}}.$$
\end{lemma}
\begin{proof}
The proof of the lemma is similar to Section 8.2 in Khot and Saket \cite{KhotS2014b}. Consider any set $A \subseteq \mathcal V$ of fractional size $s$. For every $v\in V$, let $A_v:\F_2^{m\times m} \rightarrow \{0,1\}$ be the indicator function that is extended such that it is constant over cosets of $\mathcal H_v$.  $A$ is an independent set if and only if
\begin{equation}\label{eqn:indep}
\Theta :=  \E_{u,v,w} ~\E_{X_i,Y_i  \in \mathcal T_{2,8}}~\prod_{i=1}^4A_v(X_i)A_w(Y_i) = 0.
\end{equation}
Now we do the Fourier expansion and take expectations over $X_1,X_2,Y_1,Y_2$ to obtain the following:
\begin{align*}
 \Theta =  \E_{u,v,w} \sum_{\substack{\alpha_1, \alpha_2\\ \beta_1,\beta_2} \in \F_2^{m\times m}} &\E_{F,\overline x,\overline x'} \Bigg[ \widehat A_v(\alpha_1)^2 \E_{\overline y}\left[\chi_{\alpha_1}(\overline x \otimes \overline y)\right]\chi_{\alpha_1}(F\circ \pi)\\  
 & \widehat A_v(\alpha_2)^2 \E_{\overline z}\left[\chi_{\alpha_2}((\overline x +\overline e_m) \otimes \overline z)\right]\chi_{\alpha_2}(F\circ \pi)\\
 &\widehat A_w(\beta_1)^2 \E_{\overline y'}\left[\chi_{\beta_1}(\overline x'  \otimes \overline y' )\right]\chi_{\beta_1}(F\circ \sigma) \chi_{\beta_1}(\overline e_m \otimes \overline e_m)\\  
 & \underbrace{\widehat A_w(\beta_2)^2 \E_{\overline z'}\left[\chi_{\beta_2}((\overline x'+\overline e_m)  \otimes \overline z')\right]\chi_{\beta_2}(F\circ \sigma) \chi_{\beta_2}(\overline e_m \otimes \overline e_m) \Bigg]}_{=:\Term_{u,v,w}(\alpha_1,\alpha_2,\beta_1,\beta_2)}
\end{align*}

Note that since $F\in \F_2^{r\times r}$ is chosen uniformly at random,
$$ \E_F \chi_{\alpha_1}(F\circ \pi) \chi_{\alpha_2}(F\circ \pi) \chi_{\beta_1}(F \circ \sigma) \chi_{\beta_2}(F\circ \sigma) =\E_F (-1)^{\langle \pi(\alpha_1 +\alpha_2), F \rangle + \langle \sigma(\beta_1 +\beta_2), F \rangle}$$
is zero unless $\pi(\alpha_1 +\alpha_2) = \sigma(\beta_1 +\beta_2)$. Let $\nu(\alpha):= (-1)^{\langle \alpha, \overline e_m \otimes \overline e_m \rangle}$. Now taking expectations over $\overline x,\overline y, \overline z, \overline x',\overline y',\overline z'$,  and noting that $\langle\alpha, x \otimes y \rangle = \langle \alpha x, y \rangle$, we obtain 
\begin{equation}\label{eqn:sim-term}
\begin{aligned}
  \Term_{u,v,w}(\alpha_1, \alpha_2, \beta_1,\beta_2) = (-1)^{\nu(\beta_1+\beta_2)}& \widehat A_v(\alpha_1)^2 \widehat A_v(\alpha_2)^2 \widehat A_w(\beta_1)^2 \widehat A_w(\beta_2)^2\\
& \Pr_{\overline x} \left[ \alpha_1 \overline x = 0 \wedge \alpha_2 \overline x = \alpha_2 e_m  \right] \cdot\\
& \Pr_{\overline x'} \left[ \beta_1 \overline x' = 0 \wedge \beta_2 \overline x' = \beta_2 e_m  \right]
\end{aligned}
\end{equation}

when $\pi(\alpha_1 +\alpha_2) = \sigma(\beta_1 +\beta_2)$ and $0$ otherwise.
Define:
\begin{align}
\Theta_0 &= \E_{u,v,w} \sum_{\substack{\rank(\alpha_1+\alpha_2), \rank(\beta_1+\beta_2) \leq k\\  \pi(\alpha_1+\alpha_2)= \sigma(\beta_1+\beta_2) \\ \nu(\beta_1+\beta_2) = 0 } }  \Term_{u,v,w}(\alpha_1, \alpha_2, \beta_1,\beta_2)\\
\Theta_1 &= \E_{u,v,w}\sum_{\substack{\rank(\alpha_1+\alpha_2), \rank(\beta_1+\beta_2) \leq k\\ \pi(\alpha_1+\alpha_2)= \sigma(\beta_1+\beta_2) \\  \nu(\beta_1+\beta_2) = 1 } }  \Term_{u,v,w}(\alpha_1, \alpha_2, \beta_1,\beta_2)\\
\Theta_2 &= \E_{u,v,w}\sum_{\substack {\max\left\{  \rank(\alpha_1+\alpha_2) ,\rank(\beta_1+\beta_2)\right\} > k\\ \pi(\alpha_1+\alpha_2)= \sigma(\beta_1+\beta_2)} }  \Term_{u,v,w}(\alpha_1, \alpha_2, \beta_1,\beta_2)
\end{align}
We lower bound $\Theta_0$ by $s^8$,  upper bound $|\Theta_1|$ by $\delta$ and $|\Theta_2|$ by $1/2^{k/2+1}$ below. Along with \eqref{eqn:indep}, this will prove \lref[Lemma]{lem:soundness-28}.

\subsubsection{Lower bound on \texorpdfstring{$\Theta_0$}{Theta0}}
Note that all terms in $\Theta_0$ are positive.
Now consider the term corresponding to $\alpha_1=\alpha_2=\beta_1=\beta_2=0$.
\begin{equation}
\E_{u,v,w}\widehat A^4_v(0) \widehat  A^4_w(0) = \E_u \left(\E_v \widehat  A^4_v(0) \right)^2 \geq \left(\E_{uv} \widehat  A_v(0)\right )^8 \geq s^8.
\end{equation}

\subsubsection{Upper bound on \texorpdfstring{$|\Theta_1|$}{Theta1}}
We can upper bound $|\Theta_1|$ by
\begin{equation}\label{eqn:fourier-decod}
\E_{u,v,w} \sum_{\substack{ \rank(\alpha_1+\alpha_2),\rank(\beta_1+\beta_2) \leq k,\\ \pi(\alpha_1+\alpha_2) = \sigma(\beta_1+\beta_2),\\ \nu(\beta_1+\beta_2)=1 }} \widehat A^2_v(\alpha_1) \widehat A^2_v(\alpha_2) \widehat A_w^2(\beta_1) \widehat A_w^2(\beta_2).
\end{equation}
Consider the following strategy for labeling vertices $u\in U$ and $v\in V$. For $u\in U$, pick a random neighbor $v$, choose $(\alpha_1,\alpha_2)$ with probability $\widehat A^2_v(\alpha_1) \widehat A^2_v(\alpha_2)$ and set its label to $\pi(\alpha_1+\alpha_2)$. For $w\in V$, choose $(\beta_1,\beta_2)$ with probability $\widehat A_w^2(\beta_1) \widehat A_w^2(\beta_2)$ and set its label to $\beta_1+\beta_2$. Since $A_w$ is folded, both $\beta_1$ and $\beta_2$ are symmetric and satisfies $C_v$. Since these constraints are homogeneous, $\beta_1+\beta_2$ is also symmetric and satisfies $C_v$. Also $\pi$ maps symmetric matrices to symmetric matrices. Note that \eqref{eqn:fourier-decod} gives the probability that a random edge $(u,w)$ of the label cover is satisfied by this labeling. Hence \eqref{eqn:fourier-decod} and $|\Theta_1|$ are upper bounded by $\delta$.

\subsubsection{Upper bound on \texorpdfstring{$|\Theta_2|$}{Theta2}}
Note that if the $\rank(\alpha) >k$, for any fixed $b$, $\Pr_{\overline x}[\alpha x = b] \leq 1/2^{k+1}$. All terms in $\Theta_2$ has $\max \{ \rank(\alpha_1), \rank(\alpha_2),\rank(\beta_1), \rank(\beta_2) \} > k/2.$ 
From \eqref{eqn:sim-term} we have that, for any fixed choice of $u,v,w$  each term in $\Theta_2$ has absolute value at most $1/2^{k/2+1}$. Since $A,B$ are $\{0,1\}$ valued functions, sum of their squared coefficients is upper bounded by $1$ (i.e. Parseval's inequality). Thus $|\Theta_2| \leq 1/2^{k/2+1}$.

\end{proof}

\begin{proof}[Proof of Theorem \ref{thm:2c8u}]
We already saw in \lref[Section]{sec:yes-28} that an YES instance of label cover is mapped to a $2$-colorable hypergraph.
Since $k= (\log N)^{1/8 - 2\epsilon}$ and $\delta = 2^{-(\log N)^{1/4 -2\epsilon}}$, $s \leq2^{-(\log N)^{1/8 -3\epsilon}}.$ Also the number of vertices in $\mathcal G$, 
$$n \leq N 2^{m^2} \leq N \cdot 2^{(\log N)^{10/4 +2\epsilon}}.$$
 From \lref[Lemma]{lem:soundness-28} and above,  a NO instance of label cover is mapped to a hypergraph $\mathcal G$ that has no independent set of relative size $2^{-(\log n)^{1/20-4\epsilon}}$.
\end{proof}




\section{$4$-colorable $4$-uniform hypergraphs}
In this section, we modify the reduction in the previous section, so that the uniformity of the hypergraph produced is decreased to $4$ at the cost of increasing the number of colors required in the YES case to $4$. This method was proposed by Guruswami \etal \cite{GuruswamiHHSV2014}. The hypergraph $\mathcal G=(\mathcal V, \mathcal E)$ constructed will have vertices $$\mathcal V= V\times\left( \F_2^{m\times m}\times  \F_2^{m\times m}/\mathcal H_v \times \mathcal H_v\right).$$  Any $4$-coloring of $\mathcal G$ can be expressed as a collection of functions
 $$A'_v:\left( \F_2^{m\times m}\times  \F_2^{m\times m}/\mathcal H_v \times \mathcal H_v\right) \rightarrow \{0,1\}^2,  \text{ for } v \in V.$$ 
We can uniquely extend such functions to get $A_v:\F_2^{m\times m} \times \F_2^{m\times m} \rightarrow \{0,1\}^2$ which is constant over cosets of $\mathcal H_v \times \mathcal H_v$. This ensures that $A$ satisfies the following: if $\alpha = (\alpha_1,\alpha_2) \in \F_2^{m\times m} \times \F_2^{m\times m}$ is such that $\widehat A(\alpha)$ is non-zero, then $\alpha_1,\alpha_2$ are both symmetric and satisfies $C_v$. The set of edges $\mathcal E$ will be defined by the test mentioned below.

\paragraph{$4$-Colorable $4$-Uniform Test}
\begin{enumerate}
\item  Sample $v,w$ and $\{ X_i ,Y_i \}_{i=1}^4$ from the distribution $\mathcal T_{2,8}$ as described by the test in the previous section.

\item Accept if and only if the following $4$ values are not all equal :
\begin{align*}
A_v(X_1,X_2)~~~A_v(X_3, X_4)~~~A_w(Y_1,Y_2)~~~A_w(Y_3, Y_4)
\end{align*}

\end{enumerate}
\subsection{YES Case}\label{sec:yes-44}
Given a perfectly satisfying labeling $\overline y_v \otimes \overline y_v$ for $v\in V$ and $\overline x_u \otimes \overline x_u$ for $u\in U$, we define the following $4$-coloring for $\mathcal G$:  for each $v \in V$,  $$A_v(X_1,X_2):= \left(\langle X_1,\overline y_v \otimes \overline y_v \rangle,\langle X_2,\overline y_v \otimes \overline y_v \rangle \right).$$ Note that such a function is constant over cosets of $\mathcal H_v$.  Using the arguments from \lref[Section]{sec:yes-28}, it is easy to see that $A$ is a valid $4$-coloring of $\mathcal G$.
\subsection{NO Case}\label{sec:no-44}
The analysis of the NO case is similar to \lref[Section]{sec:no-28}.

\ifcomment

\section{$2$-colorable $4$-uniform Hypergraphs}
\girish{There are gaps in the analysis that we hope to resolve soon. Any suggestions are welcome. :-)}
In this section, we propose a reduction to $4$-uniform hypergraphs, where the YES case is $2$-colorable. As described in the previous sections, the reduction is defined by a test on functions $A_v:\F_2^{m\times m} \rightarrow \{0,1\}$ that is folded over cosets of $\mathcal H_v$  for each $v$ in the ``big" side of the label cover instance. 
\paragraph{$2$-Colorable $4$-Uniform  Test $\mathcal T_{2,4}$}
\begin{enumerate}
\item Choose $u \in U$ uniformly at random and $v,w \in V$ uniformly and independently at random from the neighbors of $u$. Let $\pi_1,\pi_2:\F_2^{m \times m} \rightarrow \F_2^{r \times r}$ be the projections corresponding to the edges $(u,v),(u,w)$ respectively and $\sigma_1,\sigma_2:\F_2^m \rightarrow \F_2^r$ be the linear operators obtained from Theorem 7.1 in \cite{KhotS2014b} such that $\pi_i(\alpha):= \rho_i \alpha \rho_i^T$. Uniformly and independently at random choose $X_1, Y_1 \in \F_2^{m \times m}$ and $\overline x \in \F_2^r, \overline y, \overline z \in \F_2^m$. Let $\overline e_m \in \F_2^m$ be the vector with only the $m$\th entry $1$ and the rest is $0$. 
\item Accept if and only if the following are not all equal:
\begin{align*}
&A_v(X_1)&&A_v(X_2)&\text{ where } X_2 &:= X_1+ (\sigma_1^T \overline x) \otimes \overline y + \overline e_m \otimes \overline e_m\\
&A_w(Y_1)&&A_w(Y_2)&\text{ where } Y_2 &:= Y_1+( \sigma_2^T \overline x+ \overline e_m) \otimes \overline z  + \overline e_m \otimes \overline e_m
\end{align*}
\end{enumerate}

\subsection{YES case}
Let  $\overline y_v \otimes \overline y_v$ for $v\in V$ and $\overline x_u \otimes \overline x_u$ for $u\in U$ be a perfectly satisfying labeling of the label cover instance. Also the $m$\th coordinate of $\overline y_v$ is $1$. Consider the following $2$-coloring of $\mathcal G$: for each $v \in V$,  $A_v(X):= \langle X,\overline y_v \otimes \overline y_v \rangle$. Note that such a function is constant over cosets of $\mathcal H_v$.  Let 
\begin{align*}
&x_1 := \langle X_1 , \overline y_v \otimes \overline y_v \rangle& ~~~~~&y_1 := \langle Y_1 , \overline y_w \otimes \overline y_w \rangle.
\end{align*}
The assignments to the $4$ query locations are:
\begin{align*}
&x_1& ~~~~~&x_1+\langle \overline x_u, \overline x\rangle \langle \overline y_v, \overline y \rangle + 1\\
&y_1& ~~~~~&y_1+ (\langle \overline x_u, \overline x\rangle +1) \langle\overline y_w, \overline z \rangle  +  1
\end{align*}
It is easy to see at least one of the $2$ rows are always not equal. Hence $A$ is a valid $2$-coloring of $\mathcal G$.

\subsection{NO Case}

Suppose the reduction was applied to a NO instance of label cover. Let $k$ and $\delta$ be the parameters specified by Theorem 7.2 in \cite{KhotS2014b}.

\begin{lemma}\label{lem:soundness-24}\girish{This lemma is not proved yet!}
If there is an independent set in $\mathcal G$ of relative size $s$ then
$$ \poly(s) \leq O(\delta) +2^{-\Omega(k)}.$$
\end{lemma}
\begin{proof}
 Consider any set $A \subseteq \mathcal V$ of fractional size $s$. For every $v\in V$, let $A_v:\F_2^{m\times m} \rightarrow \{0,1\}$ be the indicator function that is extended such that it is constant over cosets of $\mathcal H_v$.  $A$ is an independent set if and only if
\begin{equation}\label{eqn:indep-24}
\Theta :=  \E_{u,v,w} ~\E_{X_i,Y_i  \in \mathcal T_{2,4}} A_v(X_1)A_v(X_2)A_w(Y_1)A_w(Y_2) = 0.
\end{equation}
Now we do the Fourier expansion and take expectations over $X_1,Y_1$ to obtain the following:
\begin{align*}
 \Theta =  \E_{u,v,w} \sum_{\alpha, \beta \in \F_2^{m\times m}} &\E_{\overline x} \Bigg[ \widehat A_v(\alpha)^2 \E_{\overline y}\left[\chi_{\alpha}(\rho^T_1\overline x \otimes \overline y)\right]\chi_{\alpha}(\overline e_m\otimes \overline e_m)\\  
 &\underbrace{\widehat A_w(\beta)^2 \E_{\overline z}\left[\chi_{\beta}((\rho^T_2\overline x +\overline e_m)   \otimes \overline z )\right] \chi_{\beta}(\overline e_m \otimes \overline e_m) \Bigg]}_{=:\Term_{u,v,w}(\alpha,\beta)}
\end{align*}

Taking expectations over $\overline x, \overline y, \overline z$, we get
\begin{equation}\label{eqn:sim-term}
\begin{aligned}
  \Term_{u,v,w}(\alpha,  \beta) = (-1)^{\nu(\alpha+\beta)} \widehat A_v(\alpha)^2  \widehat A_w(\beta)^2 \Pr_{\overline x} \left[ (\rho_1 \alpha)^T \overline x = 0 \wedge  (\rho_2\beta)^T \overline x = \beta^T \overline e_m  \right]
\end{aligned}
\end{equation}
Notice that terms with $\nu(\alpha +\beta)=0$ are all positive and 
\begin{equation}\label{eqn:large-rank-24}
\max \{ \rank(\rho_1 \alpha) ,  \rank(\rho_2 \beta) \} > k  \Rightarrow|\Term_{u,v,w}(\alpha,\beta)| \leq 1/2^{k+1}.
\end{equation}

Now consider the following expectation
\begin{align*}
\Theta_v &:= \E_{X,\overline x, \overline y} A_v(X) A_v\left(X + (\sigma^T \overline x) \otimes \overline y + \overline e_m \otimes \overline e_m \right)\\
&= \sum_\alpha  (-1)^{\nu(\alpha)}\widehat A_v(\alpha)^2 \E_{\overline x, \overline y}\chi_\alpha \left(\sigma^T \overline x \otimes \overline y\right)\\
&=\sum_\alpha  \underbrace{(-1)^{\nu(\alpha)}\widehat A_v(\alpha)^2 \Pr_{\overline x}\left[(\sigma \alpha)^T \overline x = 0 \right]}_{=: \Term_v(\alpha)}
\end{align*}
\girish{We conjecture the following which is similar to Lemma 3.2 in \cite{Saket2013}.}
\begin{conjecture}
If $\E A_v = s_v$ then
$$\Theta_v \geq \poly(s_v)$$
\end{conjecture}
If the conjecture is true, we have that 
\begin{equation}\label{eqn:uncorrelated}
\Theta':= \E_{u,v,w} \Theta_v \Theta_w \geq \poly(s).
\end{equation}
Now consider, 
\begin{align*}
\Theta' - \Theta = \E_{u,v,w} \sum_{\alpha,\beta} \underbrace{\left(  \Term_v(\alpha)\Term_w(\beta)  - \Term_{u,v,w}(\alpha,\beta)\right)}_{:=\Term'_{u,v,w}(\alpha,\beta)} \geq \poly(s)
\end{align*}
since $\Theta = 0$ from \eqref{eqn:indep-24}. Now define:
\begin{align}
\Theta_0 &= \E_{u,v,w}\sum_{\substack{(\rank(\alpha) \leq k \wedge \rank(\rho_1\alpha) \leq \sqrt k) \wedge \\  (\rank(\beta) \leq k \wedge \rank(\rho_2\beta) \leq \sqrt k)} }  \Term'_{u,v,w}(\alpha, \beta)\\
\Theta_1 &= \E_{u,v,w}\sum_{\substack{(\rank(\alpha) \geq k \wedge \rank(\rho_1\alpha) \leq\sqrt k) \vee \\  (\rank(\beta) \geq k \wedge \rank(\rho_2\beta) \leq \sqrt k)} }  \Term'_{u,v,w}(\alpha, \beta)\\
\Theta_2 &= \E_{u,v,w}\sum_{\substack {\max\left\{  \rank(\rho\alpha) ,\rank(\rho\beta)\right\} > \sqrt k} }  \Term'_{u,v,w}(\alpha, \beta)
\end{align}
From \eqref{eqn:large-rank-24}, $|\Theta_2| \leq 2^{-\sqrt k }$.

\begin{conjecture}
$|\Theta_1|$ is small similar to property (b) in Theorem 2.4 in \cite{Saket2013}.
\end{conjecture}

\begin{conjecture}
If $|\Theta_0| \geq \poly(s)$  then there is a labeling to the label cover instance that satisfies $\poly(s)$ fraction of its constraints.
\end{conjecture}
\end{proof}
\fi

\ifcomment
\section{$3$-colorable $3$-uniform Hypergraphs}

In this section, we give a reduction to $3$-uniform hypergraphs, where the YES case is $3$-colorable. We will be using the multi-layered version of the label cover over $\F_3$ for this reduction. As before,
we will construct a hypergraph $\mathcal G=(\mathcal V, \mathcal E)$. For $v = (u_1,
\cdots,u_i,v_{i+1},\cdots ,v_\ell) \in V_i$, let $\mathcal H_v \subseteq \F_3^{m_i\times m_i}$ be the dual of the subspace of the set matrices in $\mathcal L_i$ that are symmetric and diagonal blocks corresponding to $v_{i+1},\cdots, v_\ell$ satisfies the constraints $C_{v_i}$'s. The reduction is given by a test on functions $A_v: \F_3^{m_i \times m_i} \rightarrow \{0,1\}$ folded over $\mathcal H_v$ for $v\in V_i, i \in \{0,\cdots, \ell -1\}$.  
\paragraph{$3$-Colorable $3$-Uniform  Test $\mathcal T_{3,3}$}
\begin{enumerate}
\item Choose two layers $0\leq i < j < \ell$ uniformly at random and then choose $(u,v) \in E_{ij}$. Let $m:= m_i, r:= m_j$ and $\pi:\F_3^{m \times m} \rightarrow \F_3^{r \times r}$ be the projection corresponding to the edges $(u,v)$. Uniformly and independently at random choose $X \in \F_3^{r\times r}, Y \in \F_3^{m \times m}$ and $\overline y \in \F_3^m$. Let $\overline e_m \in \F_3^m$ be the vector with only the $m$\th entry $1$ and the rest is $0$. 
\item Accept if and only if the following are not all equal:
\begin{align*}
&A_u(X)&&B_v(Y)&&B_v(Y') &\text{ where } Y' &:= \overline y \otimes \overline y + \overline e_m \otimes \overline e_m - Y - X\circ \pi
\end{align*}
\end{enumerate}

\subsection{YES Case}\label{sec:yes-33}
Let  $\overline x_v \otimes \overline x_v$ for $v\in V_i, i \in \{0,\cdots,\ell-1\}$ perfectly satisfying labeling of the label cover instance. Consider the following $3$-coloring of $\mathcal G$: for each $v \in V$,  $A_v(X):= \langle X,\overline y_v \otimes \overline y_v \rangle$. Note that such a function is constant over cosets of $\mathcal H_v$.  Let 
\begin{align*}
&x_1 := \langle X_1 , \overline y_v \otimes \overline y_v \rangle& ~~~~~&x_2 := \langle X_2 , \overline y_v \otimes \overline y_v \rangle\\
&y_1 := \langle Y_1 , \overline y_w \otimes \overline y_w \rangle& ~~~~~&y_2 := \langle Y_2 , \overline y_w \otimes \overline y_w \rangle
\end{align*}
and $f := \langle F, \overline x_u \otimes \overline x_u \rangle$.  Note that $\langle F\circ \pi_{u,v}, \overline y_v\otimes \overline y_v \rangle = \langle F, \pi_{u,v}(y_v \otimes y_v)\rangle = \langle F, x_u \otimes x_u\rangle$ , and $\langle \overline e_m \otimes \overline e_m ,\overline y_v \otimes \overline y_v \rangle = \langle \overline e_m , \overline y_v \rangle = 1$. Using these, the assignments to the $8$ query locations are:
\begin{align*}
&x_1& ~~~~~&x_1+\langle \overline y_v, \overline x\rangle \langle \overline y_v, \overline y \rangle + f\\
&x_2& ~~~~~&x_2 + \left(\langle\overline y_v, \overline x\rangle +1\right) \langle \overline y_v, \overline z\rangle + f \\
&y_1& ~~~~~&y_1+\langle \overline y_w, \overline x'\rangle \langle\overline y_w, \overline y' \rangle  + f + 1\\
&y_2& ~~~~~&y_2 + \left(\langle\overline y_w, \overline x'\rangle +1\right) \langle\overline y_w, \overline z'\rangle + f + 1
\end{align*}
It is easy to see at least one of the $4$ rows are always not equal. Hence $A$ is a valid $2$-coloring of $\mathcal G$.
 
\subsection{NO Case}\label{sec:no-28}
Suppose the reduction was applied to a NO instance of label cover. Let $k$ and $\delta$ be the parameters specified by Theorem 7.2 in \cite{KhotS2014b}.

\begin{lemma}\label{lem:soundness-28}
If there is an independent set in $\mathcal G$ of relative size $s$ then
$$ s^8 \leq \delta +\frac{1}{2^{k/2+1}}.$$
\end{lemma}
\begin{proof}
The proof of the lemma is similar to Section 8.2 in \cite{KhotS2014b}. Consider any set $A \subseteq \mathcal V$ of fractional size $s$. For every $v\in V$, let $A_v:\F_2^{m\times m} \rightarrow \{0,1\}$ be the indicator function that is extended such that it is constant over cosets of $\mathcal H_v$.  $A$ is an independent set if and only if
\begin{equation}\label{eqn:indep}
\Theta :=  \E_{u,v,w} ~\E_{X_i,Y_i  \in \mathcal T_{2,8}}~\prod_{i=1}^4A_v(X_i)A_w(Y_i) = 0.
\end{equation}
Now we do the Fourier expansion and take expectations over $X_1,X_2,Y_1,Y_2$ to obtain the following:
\begin{align*}
 \Theta =  \E_{u,v,w} \sum_{\substack{\alpha_1, \alpha_2\\ \beta_1,\beta_2} \in \F_2^{m\times m}} &\E_{F,\overline x,\overline x'} \Bigg[ \widehat A_v(\alpha_1)^2 \E_{\overline y}\left[\chi_{\alpha_1}(\overline x \otimes \overline y)\right]\chi_{\alpha_1}(F\circ \pi)\\  
 & \widehat A_v(\alpha_2)^2 \E_{\overline z}\left[\chi_{\alpha_2}((\overline x +\overline e_m) \otimes \overline z)\right]\chi_{\alpha_2}(F\circ \pi)\\
 &\widehat A_w(\beta_1)^2 \E_{\overline y'}\left[\chi_{\beta_1}(\overline x'  \otimes \overline y' )\right]\chi_{\beta_1}(F\circ \sigma) \chi_{\beta_1}(\overline e_m \otimes \overline e_m)\\  
 & \underbrace{\widehat A_w(\beta_2)^2 \E_{\overline z'}\left[\chi_{\beta_2}((\overline x'+\overline e_m)  \otimes \overline z')\right]\chi_{\beta_2}(F\circ \sigma) \chi_{\beta_2}(\overline e_m \otimes \overline e_m) \Bigg]}_{=:\Term_{u,v,w}(\alpha_1,\alpha_2,\beta_1,\beta_2)}
\end{align*}

Note that since $F\in \F_2^{r\times r}$ is chosen uniformly at random,
$$ \E_F \chi_{\alpha_1}(F\circ \pi) \chi_{\alpha_2}(F\circ \pi) \chi_{\beta_1}(F \circ \sigma) \chi_{\beta_2}(F\circ \sigma) =\E_F (-1)^{\langle \pi(\alpha_1 +\alpha_2), F \rangle + \langle \sigma(\beta_1 +\beta_2), F \rangle}$$
is zero unless $\pi(\alpha_1 +\alpha_2) = \sigma(\beta_1 +\beta_2)$. Let $\nu(\alpha):= (-1)^{\langle \alpha, \overline e_m \otimes \overline e_m \rangle}$. Now taking expectations over $\overline x,\overline y, \overline z, \overline x',\overline y',\overline z'$,  and noting that $\langle\alpha, x \otimes y \rangle = \langle \alpha x, y \rangle$, we obtain 
\begin{equation}\label{eqn:sim-term}
\begin{aligned}
  \Term_{u,v,w}(\alpha_1, \alpha_2, \beta_1,\beta_2) = (-1)^{\nu(\beta_1+\beta_2)}& \widehat A_v(\alpha_1)^2 \widehat A_v(\alpha_2)^2 \widehat A_w(\beta_1)^2 \widehat A_w(\beta_2)^2\\
& \Pr_{\overline x} \left[ \alpha_1 \overline x = 0 \wedge \alpha_2 \overline x = \alpha_2 e_m  \right] \cdot\\
& \Pr_{\overline x'} \left[ \beta_1 \overline x' = 0 \wedge \beta_2 \overline x' = \beta_2 e_m  \right]
\end{aligned}
\end{equation}

when $\pi(\alpha_1 +\alpha_2) = \sigma(\beta_1 +\beta_2)$ and $0$ otherwise.
Define:
\begin{align}
\Theta_0 &= \E_{u,v,w} \sum_{\substack{\rank(\alpha_1+\alpha_2), \rank(\beta_1+\beta_2) \leq k\\  \pi(\alpha_1+\alpha_2)= \sigma(\beta_1+\beta_2) \\ \nu(\beta_1+\beta_2) = 0 } }  \Term_{u,v,w}(\alpha_1, \alpha_2, \beta_1,\beta_2)\\
\Theta_1 &= \E_{u,v,w}\sum_{\substack{\rank(\alpha_1+\alpha_2), \rank(\beta_1+\beta_2) \leq k\\ \pi(\alpha_1+\alpha_2)= \sigma(\beta_1+\beta_2) \\  \nu(\beta_1+\beta_2) = 1 } }  \Term_{u,v,w}(\alpha_1, \alpha_2, \beta_1,\beta_2)\\
\Theta_2 &= \E_{u,v,w}\sum_{\substack {\max\left\{  \rank(\alpha_1+\alpha_2) ,\rank(\beta_1+\beta_2)\right\} > k\\ \pi(\alpha_1+\alpha_2)= \sigma(\beta_1+\beta_2)} }  \Term_{u,v,w}(\alpha_1, \alpha_2, \beta_1,\beta_2)
\end{align}
We lower bound $\Theta_0$ by $s^8$,  upper bound $|\Theta_1|$ by $\delta$ and $|\Theta_2|$ by $1/2^{k/2+1}$ below. Along with \eqref{eqn:indep}, this will prove \lref[Lemma]{lem:soundness-28}.

\subsubsection{Lower bound on $\Theta_0$}
Note that all terms in $\Theta_0$ are positive.
Now consider the term corresponding to $\alpha_1=\alpha_2=\beta_1=\beta_2=0$.
\begin{equation}
\E_{u,v,w}\widehat A^4_v(0) \widehat  A^4_w(0) = \E_u \left(\E_v \widehat  A^4_v(0) \right)^2 \geq \left(\E_{uv} \widehat  A_v(0)\right )^8 \geq s^8.
\end{equation}

\subsubsection{Upper bound on $|\Theta_1|$}
We can upper bound $|\Theta_1|$ by
\begin{equation}\label{eqn:fourier-decod}
\E_{u,v,w} \sum_{\substack{ \rank(\alpha_1+\alpha_2),\rank(\beta_1+\beta_2) \leq k,\\ \pi(\alpha_1+\alpha_2) = \sigma(\beta_1+\beta_2),\\ \nu(\beta_1+\beta_2)=1 }} \widehat A^2_v(\alpha_1) \widehat A^2_v(\alpha_2) \widehat A_w^2(\beta_1) \widehat A_w^2(\beta_2).
\end{equation}
Consider the following strategy for labeling vertices $u\in U$ and $v\in V$. For $u\in U$, pick a random neighbor $v$, choose $(\alpha_1,\alpha_2)$ with probability $\widehat A^2_v(\alpha_1) \widehat A^2_v(\alpha_2)$ and set its label to $\pi(\alpha_1+\alpha_2)$. For $w\in V$, choose $(\beta_1,\beta_2)$ with probability $\widehat A_w^2(\beta_1) \widehat A_w^2(\beta_2)$ and set its label to $\beta_1+\beta_2$. Since $A_w$ is folded, both $\beta_1$ and $\beta_2$ are symmetric and satisfies $C_v$. Since these constraints are homogeneous, $\beta_1+\beta_2$ is also symmetric and satisfies $C_v$. Also $\pi$ maps symmetric matrices to symmetric matrices. Note that \eqref{eqn:fourier-decod} gives the probability that a random edge $(u,w)$ of the label cover is satisfied by this labeling. Hence \eqref{eqn:fourier-decod} and $|\Theta_1|$ are upper bounded by $\delta$.

\subsubsection{Upper bound on $|\Theta_2|$}
Note that if the $\rank(\alpha) >k$, for any fixed $b$, $\Pr_{\overline x}[\alpha x = b] \leq 1/2^{k+1}$. All terms in $\Theta_2$ has $\max \{ \rank(\alpha_1), \rank(\alpha_2),\rank(\beta_1), \rank(\beta_2) \} > k/2.$ 
From \eqref{eqn:sim-term} we have that, for any fixed choice of $u,v,w$  each term in $\Theta_2$ has absolute value at most $1/2^{k/2+1}$. Since $A,B$ are $\{0,1\}$ valued functions, sum of their squared coefficients is upper bounded by $1$ (i.e. Parseval's inequality). Thus $|\Theta_2| \leq 1/2^{k/2+1}$.

\end{proof}

\begin{proof}[Proof of Theorem \ref{thm:2c8u}]
We already saw in \lref[Section]{sec:yes-28} that an YES instance of label cover is mapped to a $2$-colorable hypergraph.
Since $k= (\log N)^{1/8 - 2\epsilon}$ and $\delta = 2^{-(\log N)^{1/4 -2\epsilon}}$, $s \leq2^{-(\log N)^{1/8 -3\epsilon}}.$ Also the number of vertices in $\mathcal G$, 
$$n \leq N 2^{m^2} \leq N \cdot 2^{(\log N)^{10/4 +2\epsilon}}.$$
 From \lref[Lemma]{lem:soundness-28} and above,  a NO instance of label cover is mapped to a hypergraph $\mathcal G$ that has no independent set of relative size $2^{-(\log n)^{1/20-4\epsilon}}$.
\end{proof}

\fi

\ifcomment

\section{Old}

\begin{lemma}\label{lem:strong-smoothness}
For any $v \in V$ and symmetric matrix $\alpha \in \F_2^{m \times m}$  such that $k:=\rank(\alpha)$ and $k'_{u,v} := \rank(\pi^{(u,v)}\alpha)$
$$\Pr_{u \in N(v)}[ k'_{u,v} < k ] \leq 2^k\delta/2.$$
\end{lemma}
\begin{proof}
\begin{claim}
For a symmetric matrix $\alpha \in \F_2^{m\times m}$ and a number $k \leq m$, the following conditions are equivalent:
\begin{enumerate}
\item $ \rank(\alpha) = k$,
\item there are linearly independent $z_1, \cdots, z_k \in \F_2^m$ and numbers $s,t \geq 0$ with $k=s+2t$ such that
$$ \alpha = \sum_{i=1}^s z_i \otimes z_i + \sum_{j=1}^t ( z_{s+2j}\otimes z_{s+2j-1} +z_{s+2j-1}\otimes z_{s+2j}).$$
\end{enumerate}
 \end{claim}
 \begin{proof}
 $1 \Rightarrow 2$ is already proved in [KhotSaketb], Lemma 2.1. To see the reverse direction, assume that $\alpha$ has the said decomposition in terms of outer products of linearly independent $z_1 \cdots z_k$. It is easy to  come up with vectors $x_1,\cdots,x_k$ such that $\langle x_i , z_j\rangle = 1[i=j]$. Notice that $\forall i \in [s],~\alpha x_i = z_i$,  for $i = s+2j,~\alpha x_i = z_{s+2j-1}$ and for $i = s+2j-1, \alpha x_i = z_{s+2j}$. Hence $\rank(\alpha) = k$.
 \end{proof}
Fix an edge $(u,v)$. Suppose $\rank(\alpha)=k$ and $\rank(\pi \alpha) < k$. Then from the claim above, there are some linearly independent $z_1,\cdots z_k \in \F_2^m$ such that

$$ \alpha = \sum_{i=1}^s z_i \otimes z_i + \sum_{j=1}^t ( z_{s+2j}\otimes z_{s+2j-1} +z_{s+2j-1}\otimes z_{s+2j}).$$

Let $\rho:\F_2^m \rightarrow \F_2^r$ be the matrix of the linear operator such that $\pi \alpha = \rho \alpha \rho^T$ as given in the proof of Theorem 7.2 from Theorem 7.1 in [KhotSaketb] and $x_i = \rho z_i$. Then nce $z_i$ are linearly independent and $a_i$ is a non zero linear combination $y_u:=\sum_i a_i z_i \neq 0$. Since $y_u$ can take only $2^k-1$ values when $u$ is varied over neighbors of $v$. For each $y\neq 0$, from the smoothness condition in Theorem 7.1, $\Pr_u[\rho_u y = 0 ] \leq \delta/2$. Hence by union bound we have that
$\Pr_u[ k_{u,v} < k] \leq \Pr_u[ \rho_u y_u = 0 ] \leq (2^k-1)\delta/2$
$$ \Pr_u [y_u = y^*] \Pr_u [ \rho_u y_u = 0] \leq  \Pr_u[\rho y* = 0] \leq \delta/2.$$

\end{proof}

From Equation  \ref{eqn:fourier-exp} and Lemma \ref{lem:strong-smoothness}, we have that
\begin{align}
\E_{u,v,w}  \sum_{\substack{\alpha,\beta \text{ symmetric },\\ \rank(\alpha),\rank(\beta) \leq r } } 1[\substack{\rank(\rho\alpha) = \rank(\alpha) \leq r \wedge \\ \rank(\rho'\beta)= \rank(\beta) \leq r }]\Term_{u,v,w}(\alpha,\beta) -\delta/2 + \sum_{\substack{\alpha,\beta \text{ symmetric },\\ \rank(\alpha),\rank(\beta) > r } } \Term_{u,v,w}(\alpha,\beta) \leq 0
\end{align}

Now lets define

\begin{align}
\theta_0 &:= \E_{u,v,w} \sum_{\substack{\alpha,\beta \text{ symmetric },\\ \rank(\alpha),\rank(\beta) \leq k\\ \nu(\alpha+\beta)=0}} 1[\substack{\rank(\rho\alpha) = \rank(\alpha) \leq r \wedge \\ \rank(\rho'\beta)= \rank(\beta) \leq r }]\Term_{u,v,w}(\alpha,\beta)\\
\theta_1 &:= \E_{u,v,w} \sum_{\substack{\alpha,\beta \text{ symmetric },\\ \rank(\alpha),\rank(\beta) \leq k\\ \nu(\alpha+\beta)=1}} 1[\substack{\rank(\rho\alpha) = \rank(\alpha) \leq r \wedge \\ \rank(\rho'\beta)= \rank(\beta) \leq r }] \Term_{u,v,w}(\alpha,\beta)\\
\theta_2 &:= \E_{u,v,w} \sum_{\substack{\alpha,\beta \text{ symmetric },\\   (k< \rank(\beta) \leq r) ~ \vee ~ (k< \rank(\beta)\leq  r)  }} 1[\substack{\rank(\rho\alpha) = \rank(\alpha) \leq r \wedge \\ \rank(\rho'\beta)= \rank(\beta) \leq r }] \Term_{u,v,w}(\alpha,\beta)\\
\theta_3 &:= \E_{u,v,w} \sum_{\substack{\alpha,\beta \text{ symmetric },\\   (r< \rank(\beta) ) ~ \vee ~ (r< \rank(\beta))  }} 1[\substack{\rank(\rho\alpha)  \leq k \vee \\ \rank(\rho'\beta)\leq k }] \Term_{u,v,w}(\alpha,\beta)\\
\theta_4 &:= \E_{u,v,w} \sum_{\substack{\alpha,\beta \text{ symmetric },\\   (r< \rank(\beta) ) ~ \vee ~ (r< \rank(\beta))  }} 1[\substack{\rank(\rho\alpha)  > k \vee \\ \rank(\rho'\beta)> k }] \Term_{u,v,w}(\alpha,\beta)
\end{align}

\begin{claim}
$\theta_0 \geq s^4$
\end{claim}
\begin{proof}
Since $\nu(\alpha+\beta)=0$, all terms in $\theta_0$ are positive and in particular it contains the term $\alpha=\beta=0$ which is $s^4$. 
\end{proof}

\begin{claim}
$|\theta_1| \leq 2\delta$
\end{claim}
\begin{proof}
$$|\theta_1| \leq \E_{u,v,w} 2\sum_{\substack{\alpha,\beta \text{ symmetric },\\ \rank(\alpha),\rank(\beta) \leq k\\ \nu(\alpha)=1}} A_v^2(\alpha) A_w^2(\beta)$$
We can do the Fourier decoding argument and bound the above by the soundness $\delta$ of the Label cover.
\end{proof}

\begin{claim}
$|\theta_2|,|\theta_4| \leq  1/2^k$
\end{claim}
\begin{proof}
Since 
$$\Pr\left[ \overline x \in \ker((\alpha \sigma)^*) \cap \left(  \beta e_m + \ker((\beta \sigma')^*)\right) \right] \leq \frac{1}{2^{\max \{ \rank(\alpha \sigma^*), \rank(\beta \sigma^*) \}}}$$
the claim follows.
\end{proof}

\begin{conjecture}
The outer verifier satisfies the additional property that $\Pr_{u \in N(v)}[\rank(\rho\alpha) < k] \leq O(\delta)$ for matrix $\alpha$ with rank at least $r$. Then
$|\theta_4|\leq O(\delta)$ 
\end{conjecture}

Hence we have that
$$ s^4- 2\delta - 1/2^k \leq \delta/2.$$
Parameters are set exactly as in [KhotSaketb] to get the hardness result.
\fi
\section{Acknowledgements}
The author would like to thank Prahladh Harsha, Subhash Khot, Rishi Saket and Srikanth Srinivasan for helpful discussions. 

{\small
\bibliographystyle{prahladhurl}
\bibliography{jrnl-names-abb,prahladhbib,crossref,doc}
}

\end{document}